\documentclass[11pt]{article}
\usepackage{graphicx}

\usepackage{amsmath,fullpage}

\newtheorem{theorem}{Theorem}[section]

\newtheorem{lemma}[theorem]{Lemma}

\newtheorem{corollary}[theorem]{Corollary}

\newcommand{\eps}{\varepsilon}
\newenvironment{proof}[1][Proof]{\textbf{#1.} }{\ \rule{0.5em}{0.5em}}

\begin{document}

\title{A lower bound for online rectangle packing\thanks{Partially supported by a grant from GIF - the
German-Israeli Foundation for Scientific Research and Development
(grant number I-1366-407.6/2016).}}

\author{Leah
Epstein\thanks{ Department of Mathematics, University of Haifa,
Haifa, Israel. \texttt{lea@math.haifa.ac.il}.}}

\date{}

\maketitle

\begin{abstract}
We slightly improve the known lower bound on the asymptotic
competitive ratio for online bin packing of rectangles. We present
a complete proof for the new lower bound, whose value is above
$1.91$.
\end{abstract}

\section{Introduction}
Bin packing \cite{JoDUGG74} is a well-studied combinatorial
optimization problem. The goal is to partition items with rational
sizes in $(0,1]$ into subsets of total sizes at most $1$, called
bins. In the online version, items are presented one by one, such
that every item is assigned irrevocably to a bin before the next
item arrives. This classic variant is also called one-dimensional
bin packing.

Rectangle packing is a generalization of bin packing where every
item is an axis parallel oriented rectangle. Each rectangle $r_i$
has a height $0 < h(r_i) \leq 1$ and a width $0 < w(r_i) \leq 1$.
The objective is to partition input rectangles into subsets, such
that every subset can be packed into a bin, where a bin is a unit
square. Packing should be done such that rectangles will not
intersect, but their boundaries can touch each other and they can
also touch the boundary of the bin. As rectangles are oriented,
they cannot be rotated. In the online variant, rectangles are
presented one by one, as in the one-dimensional version. There are
two scenarios; the one where the specific packing of a rectangle
is decided upon arrival (the position inside the bin), and the
less strict one, where the algorithm keeps subsets of rectangles
that can be packed into bins, but the exact packing can decided at
termination. Typically, positive results are proved for the first
version while negative results are proved for the second one, and
thus, all results are valid for both versions.

For an algorithm $A$ for some bin packing problem, and an input
$I$, the number of bins used by $A$ is denoted by $A(I)$. In
particular, for an optimal offline algorithm $OPT$ that receives
$I$ as a set, its cost is denoted by $OPT(I)$, and this is the
minimum number of bins required for packing $I$. The approximation
ratio, or competitive ratio if $A$ is online, for input $I$ is
$\frac{A(I)}{OPT(I)}$. The absolute approximation ratio or
absolute competitive ratio is $\sup_{I} \{\frac{A(I)}{OPT(I)}\}$,
and the asymptotic approximation ratio or asymptotic competitive
ratio $R(A)$ (which is never larger than the absolute one) is
$$R(A)=\lim\limits_{N \rightarrow \infty} \sup_{I}
\left\{\frac{A(I)}{OPT(I)} \  \bigg{|} OPT(I) \geq
N\right\}=\limsup\limits_{N \rightarrow \infty} \max_{I}
\left\{\frac{A(I)}{OPT(I)} \  \bigg{|} OPT(I)=N\right\}  .
$$

For one-dimensional online bin packing, it is known that the
asymptotic competitive ratio is in $[1.5427809,1.57828956]$
\cite{BBDEL_ESA18,BBDEL_newlb} (see also
\cite{balogh2012new,HvS16,Seiden02J,Vliet92,LeeLee85,Yao80A}).

For rectangle packing, there is a number of articles where various
algorithms are designed
\cite{coppersmith1989multidimensional,csirik1993two,Csivan93,SeidenS03,epstein2005optimal,han2011new}.
Where the last work is the one of Han et al. \cite{han2011new},
and the current best asymptotic competitive ratio is still above
$2.5$. The history of lower bounds is as follows. Galambos showed
a lower bound of $1.6$ on the asymptotic competitive ratio of any
algorithm \cite{G91}. This was improved by Galambos and van Vliet
to approximately $1.808$ by applying the same idea multiple times
\cite{GV94}.  By increasing the number of types of items in every
part of the construction, an improved lower bound of approximately
$1.851$ was shown by van Vliet \cite{vanv}. Finally, by applying
an additional modification, a lower bound of $1.907$ was claimed
\cite{Blitz,BlitzG}. For many years the lower bound of $1.907$ was
cited as an unpublished manuscript \cite{BlitzG}. This result
appears in the thesis of Blitz \cite{Blitz} that was not
accessible for many years. That thesis \cite{Blitz} contains
information that can assist in obtaining a proof, and can be seen
as guidelines for obtaining it. A manuscript was published on
arxiv with the details of an inferior result of approximately
$1.859$ \cite{HeyS17} also appearing in \cite{Blitz} with a
partial proof, where there are just nine types of items, while the
value $1.907$ was treated by many researchers as a conjecture.

The special case of rectangle packing, where all input items are
squares was studied as well
\cite{coppersmith1989multidimensional,vanv,Blitz,SeidenS03,epstein2005online,han2010note,balogh2017lower}.
For this version there is also a large gap between the lower bound
and upper bound on the asymptotic competitive ratio, where the
lower bound is approximately $1.75$ \cite{balogh2017lower}, while
the upper bound is above $2.1$ \cite{han2010note}. Another
generalization of square packing is non-oriented packing of
rectangles, where rectangles are still packed in an axis parallel
manner, but they can be rotated by $90$ degrees
\cite{fujita2002two,epstein2010two}. This version is very
different from the oriented one. For example, in the non-oriented
version, given rectangles of heights $0.66$ and widths of $0.34$,
any bin can contain at most two such items, while the non-oriented
version allows us to pack four such items into each bin.

As mentioned above, the previous lower bound on the asymptotic
competitive ratio is known as $1.907$ \cite{Blitz,BlitzG}, which
was cited multiple times, but there is no full proof of this
result. While the thesis of Blitz \cite{Blitz} has a number useful
guidelines for the proof, including the input and properties of a
certain linear program (LP) and its dual (see below), it does not
contain a complete and precise proof, and only the proof of the
lower bound $1.859$ was recovered completely \cite{HeyS17}. This
last construction is based on the nine types of items appearing in
the bottom three rows in Figure \ref{figgy}. In our work, we use
the guidelines of Blitz that were provided for an intermediate
result with $12$ item types, which was a lower bound of
approximately $1.905$ \cite{Blitz}. We modify the input by
replacing the first item with a potentially infinite sequence of
items, so instead of $12$ item types used, we have the last $11$
types, and we use a large number of types instead of the first
type. This approach allows us to provide a complete proof and show
a slightly higher lower bound of $1.9100449$ on the asymptotic
competitive ratio of any online algorithm for the packing problem
of rectangles into unit square bins. The construction of Blitz
giving a lower bound of $1.907$ consists of another row of items
on the top compared to Figure \ref{figgy}, that is, the small
empty space in the top of the bin in Figure \ref{figgy} also
contains three items. However, replacing the first item of this
construction leads to a result inferior to the one which we prove,
since the first item type  out of the $15$ has very small height,
and replacing it with a sequence of items only increases the bound
by a very small amount. On the other hand, replacing the first
item of the construction with nine item types does not increase
the lower bound above $1.9$.

We briefly discuss the relation between the proof methods. Here,
we do not provide the details of proving the results of Blitz
\cite{Blitz} using our method, since we show a better result.
However, using our proof methods it is possible to recover all
three lower bounds mentioned the thesis of Blitz \cite{Blitz} for
rectangle packing, and many of the required properties are proved
here. For comparison between the two methods (which are related),
we describe the approach of \cite{vanv,Blitz} for proving lower
bounds  on the competitive ratio for inputs of a specific form.
Such an input has several types of items fixed in advance, where
the input is of the form that at each time a large number of
identical items arrive (those are items of some type), and then
the input may be stopped (if the number of bins already used by
the algorithm is relatively high) or it may continue (if not all
item types were presented yet). Note that not all lower bounds for
the asymptotic competitive ratio of bin packing problems have this
structure, and inputs may have branching or clusters of items of
close but slightly different sizes
\cite{BBDGT,BBDEL_ESA,balogh2017lower}, though many results do
have the form we discuss here
\cite{Vliet92,FK13,BDE,balogh2012new}.

For the kind of inputs we described here, which will be used in
our construction, it is possible to analyze packing patterns. A
pattern is a multiset of items that can be packed into a bin. One
can generate all such patterns for a given input or they can be
analyzed without generating them. If the number of types is
constant as we assume here, it is possible to write an LP whose
variables are the numbers of patterns of every kind. The LP states
the relation between numbers of items and numbers of bins with all
possible patterns, i.e., numbers of items are counted as a
function of the number of patterns containing such items
(multiplied by the numbers of items in different patterns), and it
is ensured that all items are indeed packed. Patterns are
partitioned into subsets where every subset consists of patterns
whose bins are first used after the arrival of one type of items.
This is done since bins only count towards the cost of the
algorithm starting the arrival time of the first items packed into
them. Obviously, there are also constraints stating that the
competitive ratio is not violated. The inputs are sufficiently
large such that the absolute competitive ratio and the asymptotic
one are equal. The cost of the algorithm is also based on numbers
of suitable patterns, while the optimal cost is computed based on
the input. There is work where this LP is solved
\cite{yang2003ordered}, and work where the dual LP is analyzed too
\cite{vanv,Blitz,HeyS17}. In the primal LP, there are two
constraints for every item type. Thus, the dual LP has two
variables for every item type. It is frequently the case that in
an optimal solution to the dual LP the two sets of variables
differ by just a multiplicative factor.

Here, we do not use a linear program, though a linear program
corresponding to our approach would not contain a variable for
every pattern, but just a variable for the number of bins opened
by an online algorithm after the arrival of a type of items. Thus,
we would have one variable for all patterns of one subset in the
partition of patterns. We use weights for items, and these weights
are strongly related to the values of variables in solutions to
dual LP's. In our method, it is required to find an upper bound on
the total weight of any pattern in every subset of the partition,
and this is also required in the method with LS's and their dual
LP's described above. Thus, we can use the tables of data
\cite{Blitz}, both for total weights and for optimal solutions. As
these tables were given without proof, we fill this gap and
provide proofs. For our method it is not required to know the
precise values (for maximum weights and costs of optimal
solutions) but only upper bounds on these values. The method we
apply was used in the past \cite{BDE,balogh2017lower}. We use it
here as in \cite{BDE} even though the packing problem here is
different, since the method defined there can be used for many bin
packing variants.



\section{Lower bound}
Our input is based on a modification of one of the inputs of Blitz
\cite{Blitz}, where the first item is replaced with a sequence of
items. Note that this is not the input for which a lower bound of
approximately $1.907$ was claimed, but an inferior input for which
the claimed lower bound was approximately $1.905$.


Let $N>0$ be a large integer divisible by $5^k \cdot 7224$. Let $k
\geq 4$ be an integer, which is seen as a constant that is
independent of $N$. Let $\delta>0$ and $\eps>0$ be very small
values, such that $\delta < \frac 1{2^{3k+50}}$ (where in
particular, $2^{50}\delta < \frac 1{20}$) and $\eps < 0.0001$. Let
$$h_1=\frac 1{43}+\eps,  \ \ \ \ h_2= \frac 17+\eps, \ \ \ \
h_3=\frac 13+\eps, \mbox{\ \ \ \ and \ \ \ \ } h_4=\frac 12+\eps .
$$ Note that $h_1+h_2+h_3+h_4=\frac{1805}{1806}+4\eps <1$.

The input consists of $k+9$ item types. The first $k$ item types
have heights of $h_1$, the next three item types have heights of
$h_2$, the following three item types have heights of $h_3$, and
the final three item types have heights of $h_4$. The input may
stop after each one of the $k+9$ item types, and in case that some
item type is presented, there are $N$ identical items of this
type. For $i=1,\ldots,k$, the $i$th item type out of the first $k$
types is denoted by type $\ell_{1i}$. For $i=1,2,\ldots,k-2$, it
is defined by its width $$w_{1i}=\frac{1+\delta}{5^{k-i-1}},$$ and
therefore our addition to the original input \cite{Blitz} is
replacing an item whose width is just below $\frac 14$ by items
slightly wider than negative powers of $5$. We also let the widths
of type $\ell_{1(k-1)}$ be $w_{1(k-1)}=\frac {1+2^{40}\delta}4$,
and the width of type $\ell_{1k}$ is defined as $w_{1k}=\frac
{1+2^{40}\delta}2$. Dimensions for all item types are also given
in Table \ref{tabtab}, and an illustration is given in figure
\ref{figgy}. We let $h_{1i}=h_1$ for $1\leq i \leq k$, and
$h_{ji}=h_j$ for $j=2,3,4$ and $i=0,1,2$.


Thus, for $1 \leq i \leq k-3$, we have $w_{1(i+1)}=5\cdot w_{1i}$,
and we also have $w_{1k}=2\cdot w_{1(k-1)}$. For any integer $1
\leq t \leq k-2$, the total width of $t$ items, consisting of
exactly one item of every type $\ell_{1i}$ for any $1 \leq i \leq
t$, is
$$\sum_{i=1}^{t} w_{1i} = (1+\delta)\sum_{i=1}^t \frac
{1}{5^{k-i-1}}= (1+\delta)\frac{1}{5^{k-t-1}}\sum_{i=1}^t
\frac{1}{5^{t-i}}= (1+\delta)\frac{1}{5^{k-t-1}}\sum_{j=0}^{t-1}
\frac{1}{5^{j}}$$
$$=(1+\delta)\frac{1}{5^{k-t-1}}\frac{1-\frac1{5^t}}{1-\frac
15}=(1+\delta)\frac{1}{4 \cdot 5^{k-t-2}}(1-\frac
1{5^t})=(1+\delta)(\frac{1}{4 \cdot 5^{k-t-2}}-\frac{1}{4 \cdot
5^{k-2}})<\frac {1-2^{42}\delta}{4 \cdot 5^{k-t-2}}  , $$ by
$2^{42}+1<2^{43}$ and since
$$\frac{2^{43}\delta}{4 \cdot 5^{k-t-2}} < \frac{1+\delta}{4 \cdot
5^{k-2}}$$ holds by $t \leq k-2$ and $$\frac{4 \cdot 5^{k-2}}{4
\cdot 5^{k-t-2}} = 5^{t} \leq 5^{k-2} \mbox{ \ \ \ \ while \ \ \ \
} \frac{1+\delta}{2^{43}\delta} > \frac{1}{2^{43}\delta}> 2^{3k+7}
> 8^k . $$ In particular for $t=k-2$, the total width is below
$\frac 14$. Thus, we also have $$\sum_{i=1}^{k-1} w_{1i} < \frac
{1-2^{42}\delta}{4} + \frac {1+2^{40}\delta}4 < \frac 12 \mbox{ \
\ \ \ and  \ \ \ \ } \sum_{i=1}^{k} w_{1i} < \frac
{1-2^{42}\delta}{4} + 3 \cdot \frac {1+2^{40}\delta}4 < 1  . $$

The next three types are denoted by $\ell_{20}$, $\ell_{21}$,
$\ell_{22}$, and their widths are $w_{20}=\frac 14-2^{32}\delta
> \frac 15$, $w_{21}=\frac 14+2^{30}\delta$, and $w_{22}=\frac
12+2^{31}\delta$, respectively. The following three types are
denoted by $\ell_{30}$, $\ell_{31}$, $\ell_{32}$, and their widths
are $w_{30}=\frac 14-2^{22}\delta > \frac 15$, $w_{31}=\frac
14+2^{20}\delta$, and $w_{22}=\frac 12+2^{21}\delta$,
respectively. The last three types are denoted by $\ell_{40}$,
$\ell_{41}$, $\ell_{42}$, and their widths are $w_{40}=\frac
14-2^{12}\delta > \frac 15$, $w_{41}=\frac 14+2^{10}\delta$, and
$w_{42}=\frac 12+2^{11}\delta$, respectively. Note that
$w_{j0}+w_{j1}+w_{j2} <1$ and $w_{j0}+w_{j1} < \frac 12$ for
$j=2,3,4$, but

$$ w_{20} + 3 \cdot w_{1(k-1)} \geq 2 \cdot
 w_{20}
 + 2 \cdot w_{1(k-1)} \geq 3 \cdot w_{20} + w_{1(k-1)}
=3\cdot (\frac 14-2^{32}\delta)+\frac {1+2^{40}\delta}4 > 1 -
2^{34}\delta+2^{38}\delta
>1$$ and $$2 \cdot w_{20} +w_{1k} = 2 \cdot w_{20} +2 \cdot w_{1(k-1)}
>1  . $$

In addition, we have

$$ w_{(j+1)0}+3 w_{j1} \geq  2 w_{(j+1)0}+ 2 w_{j1} = 3 w_{(j+1)0}+ w_{j1} =  3(\frac 14-2^{52-10j}\delta)+(\frac 14+2^{60-10j}\delta)  >1$$ and
$$ 2 w_{(j+1)0}+ w_{j2} = 2 w_{(j+1)0}+2  w_{j1} >1 \mbox{\ \ \ \ \ for \ \ \ \ \ }
j=2,3.$$

It can be seen that $w_{20} < w_{30} < w_{40}$ while $w_{1(k-1)}
> w_{21} > w_{31} > w_{41}$ and $w_{1k)}
> w_{22} > w_{32} > w_{42}$.

\begin{figure} [h!]
\vspace{0.3cm} \hspace{0.5in}
\includegraphics[angle=0,width=\textwidth]{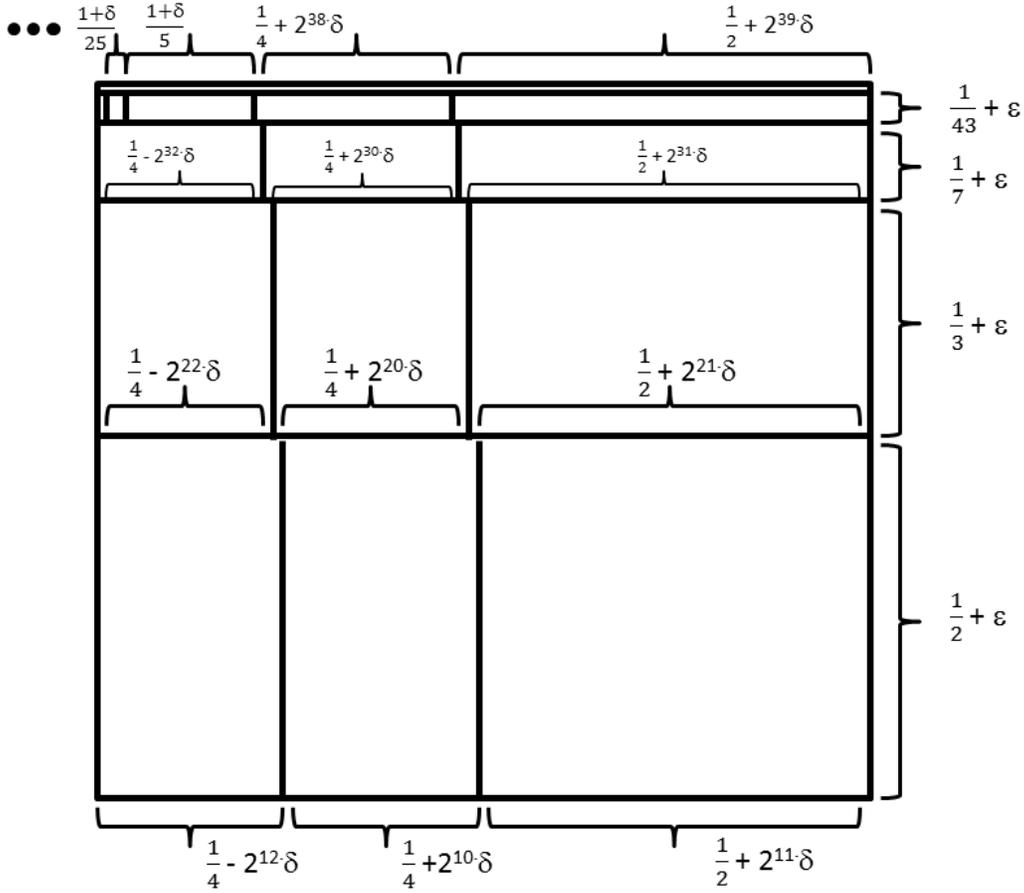}
\vspace{-0.07cm} \caption{An illustration of the input in terms of
one bin of an optimal solution for the entire input, if all item
types are presented. Item types arrive ordered from top to bottom
and from left to right. \label{figgy}}
\end{figure}

We say that an item type is later than another type if it is
presented later in the input. The weight $v_{ji}$ for an item of
type $\ell_{ji}$ is given in Table \ref{tabtab}. The weights were
selected based on dual variables provided in \cite{Blitz}.

We let $V_{ji}$ denote the maximum total weight of a bin
containing items of the types consisting of type $\ell_{ji}$ and
later types. Let $\Omega_{ji}$ be an upper bound on
$\frac{OPT_{ji}}N$, where $OPT_{ji}$ is the cost of an optimal
solution for the input up to type $\ell+{ji}$ items.  Note that
\cite{Blitz} contains tables with costs of optimal solutions and
$V_{ji}$ values for the part of the input that is identical to
ours, though it does not contain proofs of all the claimed values.

We use a theorem defined for inputs for bin packing problems, such
that the inputs have the form explained in the introduction. The
input consists of ``batches'' of identical items without
branching. Substituting our notation, the theorem states that

$$\frac{\sum_{i=1}^k v_{1i}+\sum_{j=2}^4 \sum_{i=0}^2
v_{ji}}{Q}  , $$ where $$Q=\left(\Omega_{11}\cdot
V_{11}+\sum_{i=2}^k
(\Omega_{1i}-\Omega_{1(i-1)})V_{1i}\right)+\left((\Omega_{20}-\Omega_{1k})V_{20}+
(\Omega_{21}-\Omega_{20})V_{21}+(\Omega_{22}-\Omega_{21})V_{22}\right)$$
$$+ \left(\sum_{j=3}^4 \left((\Omega_{j0}-\Omega_{(j-1)2})V_{j0} +
(\Omega_{j1}-\Omega_{j0})V_{j1}+(\Omega_{j2}-\Omega_{j1})V_{j2}\right)\right)
 ,
$$

is a lower bound on the asymptotic competitive ratio (see
\cite{BDE}).

Note that one can use an upper bound on $V_{ji}$ rather than the
actual value if all multipliers are positive, which will be the
case here. This will hold as we will ensure that the sequence of
upper bounds on $\frac{OPT_{ji}}N$ will be monotonically
non-decreasing. In fact, many of the values that we use for
$V_{ij}$ and $\Omega_{ji}$ are not just upper bounds, but they are
the precise values, though we do not prove this property and do
not use it. In order to apply the formula, one has to show that
all optimal solutions are of order of growth $\Theta(N)$ (since we
are interested in a lower bound on the asymptotic competitive
ratio), which will be shown later.

We have  $$\sum_{i=1}^k v_{1i}+\sum_{j=2}^4 \sum_{i=0}^2
v_{ji}=68.25-\frac{1}{4\cdot 5^{k-1}}  , $$ since $\sum_{i=k-1}^k
v_{1i}+\sum_{j=2}^4 \sum_{i=0}^2 v_{ji}=67$ and
$$\sum_{i=1}^{k-2} v_{1i}= \sum_{i=1}^{k-2}\frac{1}{5^{k-i-2}} =\sum_{j=0}^{k-3}
\frac{1}{5^{j}}=\frac{1-\frac{1}{5^{k-2}}}{4/5}=1.25-\frac{1}{4\cdot
5^{k-1}}  , $$  and $$Q \leq \frac{1}{168}
\cdot(42\cdot(5-\frac{1}{5^{k-3}})/5^{k-3}+\sum_{i=2}^{k-2}42\cdot
(5-\frac{1}{5^{k-i-2}})\cdot(\frac{1}{5^{k-i-2}}-\frac{1}{5^{k-(i-1)-2}})$$
$$+1\cdot 126+2\cdot 112+6\cdot 96+6\cdot 72+12\cdot 68+14\cdot
48+14\cdot 42+28\cdot 36+21\cdot 24+21\cdot 18+42\cdot 12)  . $$
Since
$$\sum_{i=2}^{k-2}(5-\frac{1}{5^{k-i-2}})\cdot(\frac{1}{5^{k-i-2}}-\frac{1}{5^{k-(i-1)-2}})=\sum_{i=2}^{k-2}(5-\frac{1}{5^{k-i-2}})\cdot\frac{4}{5^{k-i-1}}$$
$$=4\cdot\sum_{i=2}^{k-2}\frac{1}{5^{k-i-2}}-0.8 \cdot
\sum_{i=2}^{k-2}\frac{1}{25^{k-i-2}}=4\cdot
\sum_{u=0}^{k-4}\frac{1}{5^{u}}-0.8\cdot
\sum_{u=0}^{k-4}\frac{1}{25^{u}}$$
$$=4\cdot(\frac{1-(1/5)^{k-3}}{0.8})-0.8\cdot(\frac{1-(1/25)^{k-3}}{0.96})=\frac{25}{6}-\frac{1}{5^{k-4}}+
\frac{1}{6 \cdot 5^{2k-7}} \
$$

we have, $$Q \leq \frac{1}{168} \cdot
(42/5^{k-4}-42/5^{2k-6}+42(\frac{25}{6}-\frac{1}{5^{k-4}}+
\frac{1}{6 \cdot
5^{2k-7}})+5828)=\frac{6003-\frac{7}{5^{2k-6}}}{168}  .
$$

We get $$r \geq \frac{68.25-\frac{1}{4\cdot 5^{k-1}}}{Q} \geq
\frac{68.25-\frac{1}{4\cdot
5^{k-1}}}{\frac{6003}{168}-\frac{1}{24\cdot 5^{2k-6}}}  . $$
Letting $k$ grow to infinity, we get
$\frac{11466}{6003}=\frac{1274}{667}\approx 1.9100449$.

{\begin{table}[h!]
\renewcommand{\arraystretch}{1.4   }

$$
\begin{array}{||c|c|c|c|c|c||}
\hline \hline

\mbox{Item type} & \mbox{width}  & \mbox{height}  & \mbox{weight} & \mbox{upper bound}  &  \mbox{upper bound on}    \\

\mbox{$\ell_{ji}$} & w_{ji} &  h_{ji} & v_{ji} & \mbox{on\ }  V_{ji} &   168 \cdot OPT_{ji}/N \ \ \  (\Omega_{ji}) \\

\hline j=1, \ 1 \leq i \leq k-2
& \frac{1+\delta}{5^{k-i-1}} & \frac 1{43}+\eps & \frac{1}{5^{k-i-2}} & 42(5-\frac{1}{5^{k-i-2}})& \frac{1}{5^{k-i-2}}  \\

\hline j=1, \  i=k-1 & \frac
{1+2^{40}\delta}4 & \frac 1{43}+\eps & 1 & 126 & 2  \\

\hline j=1, \ i=k & \frac
{1+2^{40}\delta}2 & \frac 1{43}+\eps & 2 & 112 & 4   \\

\hline j=2, \ i=0 &   \frac 14-2^{32}\delta& \frac 1{7}+\eps & 4 & 96 & 10  \\

\hline j=2, \ i=1 &
\frac 14+2^{30}\delta & \frac 1{7}+\eps & 4 & 72 & 16  \\

\hline j=2, \ i=2
& \frac 12+2^{31}\delta & \frac 1{7}+\eps & 8 & 68 & 28  \\

\hline j=3, \ i=0
& \frac 14-2^{22}\delta & \frac 1{3}+\eps & 6 & 48 & 42  \\

\hline j=3, \ i=1
& \frac 14+2^{20}\delta & \frac 1{3}+\eps & 6 & 42 & 56  \\

\hline j=3, \ i=2
&  \frac{1}{2}+2^{21}\delta  & \frac 1{3}+\eps & 12 & 36 & 84  \\

\hline j=4, \ i=0
& \frac 14-2^{12}\delta & \frac 1{2}+\eps & 6&  24 & 105  \\

\hline j=4, \ i=1
& \frac 14+2^{10}\delta & \frac 1{2}+\eps & 6& 18& 126   \\

\hline j=4, \ i=2
&  \frac{1}{2}+2^{11}\delta & \frac 1{2}+\eps & 12& 12& 168  \\

\hline \hline

\end{array}
$$
\caption{\label{tabtab} A summary of the input and all values
required for the proof of the lower bound. The first four columns
contain definitions, and the contents of the remaining two columns
are proved in the text.}
\end{table}
}

\begin{lemma}
For every valid pair $j,i$, we have $OPT_{ji} \leq \Omega_{ji}$,
where $\Omega_{ji}$ is stated in Table \ref{tabtab}.
\end{lemma}
\begin{proof}
Let $j=1$. Let $1 \leq i \leq k$, and consider a subset of items
consisting of one item of every type $\ell_{1a}$ for $1 \leq a
\leq i$. For $1 \leq i \leq k-2$, the set has total width below
$\frac{1}{4 \cdot 5^{k-i-2}}$, and therefore one can pack them
into a rectangle of height $\frac{1}{43}+\eps$ and width
$\frac{1}{4 \cdot 5^{k-i-2}}$. A bin can be split into $42$ rows
of height $\frac 1{42}$ and ${4 \cdot 5^{k-i-2}}$ columns of width
$\frac{1}{4 \cdot 5^{k-i-2}}$, resulting in $42\cdot {4 \cdot
5^{k-i-2}}$ such rectangles. Thus, $$OPT_{ji} \leq \frac{N}{42
\cdot {4 \cdot 5^{k-i-2}}}=\frac{N}{168} \cdot \frac 1{5^{k-i-2}}
 , $$ and we let $\Omega_{1i}=\frac{1}{168}\cdot\frac
1{5^{k-i-2}}$. For $i=k-1$, the total width is below $\frac 12$,
so the columns will be of width $\frac 12$, and
$\Omega_{1k}=\frac{1}{84}$. For $i=k$, the total width is below
$1$, so the columns will be of width $1$ (that is, there are no
columns), and $\Omega_{1k}=\frac{1}{42}$.

For $j=2,3,4$, an item of type $\ell_{j0}$ can be packed into a
rectangle of width $\frac 14$ and the corresponding height ($\frac
17$ for $j=2$, $\frac 13$ for $j=3$, and $\frac 12$ for $j=1$),
two items, one of type $\ell_{j0}$ and one of type $\ell_{j1}$
have total width below $\frac 12$, and they can be packed into a
rectangle of width $\frac 12$ and the corresponding height. For
three items, one of each type out of $\ell_{j0}$, $\ell_{j1}$, and
$\ell_{j2}$, the total width is below $1$, they can be packed into
a rectangle of width $1$ and the corresponding height.

\medskip

Bounding $OPT_{20}$ is done as follows. Create $\frac{N}{24}$ bins
with six rows of height $\frac{1}7+\eps$ and six rows of height
$\frac{1}{43}+\eps$. This is possible since $\eps = 0.0001$. Every
row of height $\frac 17+\eps$ is split into four columns of width
$\frac 14$. This allows us to pack all items of type $\ell_{20}$,
as there are $24$ areas in every bin that can contain an item of
type $\ell_{20}$ each. Additionally, we can pack one item of each
type $\ell_{1i}$ (where $1 \leq i \leq k$) into every row of
height $\frac{1}{43}+\eps$, allowing to pack six items of every
such type into every bin, and leaving $\frac{3N}{4}$ items of each
such type unpacked. These items are packed into $\frac{3N}{4\cdot
42}$ additional bins, each having $42$ rows of height
$\frac{1}{43}+\eps$. Thus, we let $\Omega_{20} = \frac{10}{168}$.

\medskip

Bounding $OPT_{21}$ is done similarly, but pairs of a type
$\ell_{20}$ item and a type $\ell_{21}$ item are packed into areas
of width $\frac 12$, so they occupy $\frac{N}{12}$ bins and
$\frac{N}2$ items of each type $\ell_{1i}$ remain unpacked and
they require $\frac{N}{2\cdot 42}$ bins, leading to the definition
$\Omega_{21}=\frac{16}{168}$. For $OPT_{22}$, triples of a type
$\ell_{20}$ item, a type $\ell_{21}$ item, and a type $\ell_{22}$
item are packed into areas of width $1$, so they occupy
$\frac{N}{6}$ bins, and all items of types $\ell_{1i}$ are packed
into the rows of height $\frac 1{43}+\eps$ in the same bins. This
leads to the definition $\Omega_{22}=\frac{28}{168}$.

\medskip

Bounding $OPT_{30}$ is done as follows. Create $\frac{N}{8}$ bins
with two rows of height $\frac{1}3+\eps$, two rows of height
$\frac{1}{7}+\eps$ and two rows of height $\frac{1}{43}+\eps$.
This is possible since $\eps = 0.0001$. Every row of height $\frac
13+\eps$ is split into four columns of width $\frac 14$. This
allows us to pack all items of type $\ell_{30}$, as every bin has
eight areas where such an item can be packed. Additionally, we can
pack one item of each type $\ell_{1i}$ into every row of height
$\frac{1}{43}+\eps$, and we can pack one item of each type
$\ell_{2i}$ into every row of height $\frac{1}{7}+\eps$, leaving
$\frac{3N}{4}$ items of each such type unpacked. These items are
packed into $\frac{3N}{4\cdot 6}$ additional bins, each having six
rows of height $\frac 17+\eps$ and six rows of height
$\frac{1}{43}+\eps$. Thus, we let $\Omega_{30} = \frac{42}{168}$.

\medskip

Bounding $OPT_{31}$ is done similarly, but pairs of a type
$\ell_{30}$ item and a type $\ell_{31}$ item are packed into areas
of width $\frac 12$, so they occupy $\frac{N}{4}$ bins and
$\frac{N}2$ items of each type $\ell_{ji}$ for $j=1,2$  remain
unpacked and they require $\frac{N}{12}$ bins, leading to the
definition $\Omega_{31}=\frac{56}{168}$. For $OPT_{32}$, triples
of a type $\ell_{30}$ item, a type $\ell_{31}$ item, and a type
$\ell_{32}$ item are packed into areas of width $1$, so they
occupy $\frac{N}{2}$ bins, and all items of types $\ell_{1i}$ and
$\ell_{2t}$ are packed into the rows of height $\frac 1{43}+\eps$
in the same bins. This leads to the definition
$\Omega_{32}=\frac{84}{168}$.

\medskip

Bounding $OPT_{40}$ is done as follows. Create $\frac{N}{4}$ bins
with a row of every height out of $\frac 12+\eps$, $\frac
13+\eps$, $\frac 17+\eps$, and $\frac{1}{42}+\eps$. This is
possible since $\eps = 0.0001$. Every row of height $\frac
12+\eps$ is split into four columns of width $\frac 14$. This
allows us to pack all items of type $\ell_{40}$. Additionally, we
can pack one item of each type $\ell_{ji}$ for any $j\in\{1,2,3\}$
and any $i$ into the other rows, leaving $\frac{3N}{4}$ items of
each such type unpacked. These items are packed into
$\frac{3N}{4\cdot 2}$ additional bins, each having two rows of
every height excluding $\frac{1}{2}+\eps$. Thus, we let
$\Omega_{40} = \frac{105}{168}$.

\medskip

Bounding $OPT_{41}$ is done similarly, but pairs of a type
$\ell_{40}$ item and a type $\ell_{41}$ item are packed into areas
of width $\frac 12$, so they occupy $\frac{N}{2}$ bins and
$\frac{N}2$ items of each type $\ell_{ji}$ for $j=1,2,3$ remain
unpacked and they require $\frac{N}{4}$ bins, leading to the
definition $\Omega_{31}=\frac{126}{168}$. For $OPT_{32}$, triples
of a type $\ell_{30}$ item, a type $\ell_{31}$ item, and a type
$\ell_{32}$ item are packed into areas of width $1$, so they
occupy $N$ bins, and all items of other types are packed into the
rows of the other three heights. This leads to the definition
$\Omega_{32}=1$.
\end{proof}

We are left with the task of bounding $V_{ij}$. The bounds will be
proved using a sequence of lemmas, where the first one is general
an it is used in several proofs.

\begin{lemma}
Let $b_w$ and $b_h$ be positive integers. Consider an item size
such that the width is in $(\frac 1{b_w+1},1]$, and the height is
in $(\frac 1{b_h+1},1]$. Consider a bin that contains $f$ items of
this type (and possibly other items). Then, $f \leq b_w \cdot
b_h$.
\end{lemma}
\begin{proof}
Consider the bin and draw $b_h$ horizontal lines. Considering also
the bottom and top of the bin, the distances between any two
consecutive lines will be $\frac 1{b_h+1}$. Since the height of
the items is above $\frac 1{b_h+1}$, every item contains a part of
at least one line in its interior (such a line is not the bottom
or top). Since the width of every item is above $\frac 1{b_w+1}$,
and items cannot overlap (except for their boundary), there can be
at most $b_w$ items containing a part of a line. For every item,
associate it with a line that it contains a part of it (if there
is more than one such line, choose one arbitrarily). As there are
$b_h$ lines with at most $b_w$ items each, there are at most $b_w
\cdot b_h$ items of this type.
\end{proof}

\medskip

The last lemma shows in particular that all optimal solutions have
order of growth $\Omega(N)$, as the first $N$ items have sides
larger than $\frac{1}{5^{k-2}}$ and $\frac{1}{43}$, respectively,
so the cost of any solution is at least $\frac{N}{42 \cdot
(5^{k-2}-1)}$. An upper bound of $O(N)$ on the cost of an optimal
solution for every input follows from the total number of items
which is $(k+9)N$.

\medskip

We use the concept of dominance as in \cite{Blitz}. For an item
type $\ell_{ji}$ and an item of type $\ell_{j'i'}$, if there are
integers $c_w$ and $c_h$ such that $w_{ji} \geq c_w \cdot
w_{j'i'}$ and $h_{ji} \geq c_h \cdot h_{j'i'}$, while $v_{ji} \leq
c_w \cdot c_h \cdot v_{j'i'}$, we say that type $\ell_{j'i'}$
$(c_w,c_h)$-dominates (or simply dominates) type $\ell_{ji}$ in
the sense that in the calculation of the maximum weight of any
feasible bin, items of type $\ell_{ji}$ do not need to be
considered, as every such item can be replaced with $c_w \cdot
c_h$ items of type $\ell_{j'i'}$, without decreasing the total
weight. Note that the dominance relation is transitive. The value
$c_w \cdot c_h$ is called the factor of dominance.


\begin{lemma}
\begin{enumerate}

\item For every $j\in\{2,3,4\}$ and every $i=0,1$, type
$\ell_{ji}$ dominates $\ell_{j(i+1)}$.

\item For $i=1,2,\ldots,k-1$, type $\ell_{1i}$ dominates
$\ell_{1(i+1)}$.

\item Item type $\ell_{1(k-2)}$ dominates item type $\ell_{20}$.

\item Item type $\ell_{20}$ dominates item type $\ell_{30}$.

\item Item type $\ell_{30}$ dominates item type $\ell_{40}$.
\end{enumerate}

\end{lemma}
\begin{proof}
\begin{enumerate}
\item For every $j \in \{2,3,4\}$, the type $\ell_{j0}$ dominates
$\ell_{j1}$ since the width of the former type is smaller, and
their heights and weights are equal. Type $\ell_{j1}$ dominates
type $\ell_{j2}$ since the width of the former is twice as small,
their heights are equal, and the weight ratio satisfies
$v_{j2}/v_{j1}=2$.

\item For $1 \leq i \leq k-3$, type $\ell_{1i}$ dominates type
$\ell_{1(i+1)}$ as their heights are equal, and
$\frac{w_{1(i+1)}}{w_{1i}}=\frac{v_{1(i+1)}}{v_{1i}}=5$. Type
$\ell_{1(k-2)}$ dominated type $\ell_{1(k-1)}$ since their heights
are equal, $w_{1(k-2)}<w_{1(k-1)}$ and $v_{1(k-2)}=v_{1(k-1)}$.
Type $\ell_{1(k-1)}$ dominated type $\ell_{1k}$ since their
heights are equal, $2 \cdot w_{1(k-1)} = w_{1(k-1)}$ and $2 \cdot
v_{1(k-1)}=v_{1k}$.

\item The height of item type $\ell_{1(k-2)}$ is $\frac
1{43}+\eps$, and the height of item type $\ell_{20}$ is $\frac
17+\eps$. We have $6(\frac1{43}+\eps)<\frac 17+\eps$ as
$\eps<0.0001$. The width of item type $\ell_{1(k-2)}$ is $\frac
{1+\delta}{5}$, and the width of item type $\ell_{20}$ is $\frac
14-2^{32}\delta$. We have $\frac{1+\delta}{5}<\frac
14-2^{32}\delta$ as $2^{64}\delta <1$. As the weight of six items
of type $\ell_{1(k-2)}$ is $6$, while the weight of one item of
type $\ell_{20}$ is $4$, the domination holds.

\item Item type $\ell_{20}$ has height $\frac17+\eps$ while item
type $\ell_{30}$ has height $\frac 13+\eps$, and we have $2(\frac
17+\eps)<\frac 13+\eps$, as $\eps < 0.0001$. Item type $\ell_{20}$
has smaller width than item type $\ell_{30}$. The weight of two
items of type $\ell_{20}$ is $8$ while the weight of one type
$\ell_{30}$ item is $6$. Thus, the domination holds.

\item Item type $\ell_{30}$ has both smaller height and smaller
width than an item of type $\ell_{40}$ and they have the same
weights. Thus, the domination holds.
\end{enumerate}
\end{proof}

%
%
%
%
%
%
%
%
%
%
%
%

\begin{lemma}
In the following cases it  is sufficient to consider bins
containing only items of type $\ell_{ji}$ for the computation of
$V_{ji}$.

\begin{enumerate}

\item The case $j=4$ and $i=0,1,2$.

\item The cases $j=2,3$ and $i=0$.

\item The case $j=1$ and $i \leq k-2$. \end{enumerate}

In the cases $j=1$ and $i=k-1,k$, it is sufficient to consider
only $\ell_{1i}$ and $\ell_{20}$. In the cases $j=2$ and $i=1,2$
it is sufficient to consider only $\ell_{ji}$ and $\ell_{30}$. In
the cases $j=3$ and $i=1,2$ it is sufficient to consider only
$\ell_{ji}$ and $\ell_{40}$
\end{lemma}
\begin{proof}
The three cases where one item type can be considered follow by
transitivity of domination, since every such type dominates every
later type.

In the other cases the mentioned two item types are sufficient as
for every later type (later than $\ell_{ji}$, at least one of
these two mentioned types dominates the later one.
\end{proof}

\begin{corollary}
All $V_{ji}$ values in the table for the next cases are correct.

The case $j=4$ and $i=0,1,2$, the cases $j=2,3$ and $i=0$, and the
case $j=1$ and $i \leq k-2$.
\end{corollary}
\begin{proof}
For these values we consider one type of item.

In the cases where $i=0$ and $j\geq 2$, the width of an item is in
$(\frac 15,\frac 14]$, so $b_w=4$, and the heights for $j=2,3,4$
are in $(\frac{1}{b_h+1}, \frac{1}{b_h}]$ for $b_h=6,2,1$,
respectively. Thus, taking the item weights into account we let
$V_{20}=4\cdot 24=96$, $V_{30}=6\cdot 8=48$, and $V_{40}=6\cdot
4=24$.

As $w_{41} > \frac 14$, we let $V_{41}= 6 \cdot 3 = 18$, and as
$w_{42}
> \frac 12$, we let $V_{42}= 12$.

Consider $j=1$, for which the height is above $\frac{1}{43}$. For
$1 \leq i \leq k-2$, we let $V_{1i}=\frac{1}{5^{k-i-2}} \cdot 42
\cdot (5^{k-i-1}-1)$ since the width of these items is above
$\frac{1}{5^{k-i-1}}$ and the height is above $\frac 1{43}$.
\end{proof}

\begin{lemma}
We have $V_{31}=42$ and $V_{32}=36$.
\end{lemma}
\begin{proof}
To prove the first bound, we consider types $\ell_{31}$ and
$\ell_{40}$. Since for these two types the widths are above $\frac
15$ and the heights are above $\frac 13$, no bin can contain more
than eight such items. Moreover, for type $\ell_{40}$ the width is
above $\frac12$, so no bin can contain more than four such items.
If the bin has at most seven items, we are done, as the weight of
any item of one of these types is $6$ and therefore we assume that
there are eight such items.

Recall that we use the word {\it intersecting} with the meaning
that the intersection is in the interior and not on the boundary.
Draw two horizontal lines with distances of $\frac 13$ between
consecutive lines including the top and bottom. Due to item
heights, every item intersects at least one line (contains a part
of a line in its interior) and we associate it with such a line.
If it intersects both lines, we associate it with one of them.

As all widths are above $\frac 15$, every line can intersect at
most four items, and as there are eight items, each associated
with one of the lines, we find that every line intersects exactly
four items, and no item intersects two lines. For a given line
(one of the two), let $y_{31} \geq$ and $y_{40}\geq $ be the
(integer) numbers of items of types $\ell_{31}$ and $\ell_{40}$
associated with this line, where $\ell_{31}+\ell_{40}=4$. We have
$$1 \geq y_{31}(\frac 14+2^{20}\delta)+y_{40}(\frac
14-2^{12}\delta)
=(y_{31}+y_{40})/4+\delta(2^{20}y_{31}-2^{12}y_{40})=1+\delta(2^{20}y_{31}-2^{12}y_{40}),
$$ which implies $y_{40} \geq 2^8 y_{31}$. The only solution is
$y_{40}=4$ and $y_{31}=0$. However, this proves that the bin has
eight type $\ell_{40}$ items, a contradiction.

To prove the second bound, note that a bin can contain at most two
items of type $\ell_{32}$, as their heights are above $\frac 13$
and their widths are above $\frac 12$. Since a bin contains at
most four items of type $\ell_{40}$, if there is at most one item
of type $\ell_{32}$, we are done, as $w_{32}=12$. Assume that
there are two such items. Drawing horizontal lines as before,
every $\ell_{32}$ type item overlaps at least one such line, and
as $w_{32} > \frac 12$, each item overlaps exactly one of the
lines. Since $(\frac 12+2^{21}\delta)+2(\frac 14-2^{12}\delta)>1$,
each line overlaps at most one item of type $\ell_{40}$. Since
every type $\ell_{40}$ item overlaps at least one line, there are
at most two such items, and the total weight is at most $2\cdot
12+2\cdot 6=36$.
\end{proof}

\begin{lemma}
We have $V_{21}=72$ and $V_{22}=68$.
\end{lemma}
\begin{proof}
To prove the first bound, we consider types $\ell_{21}$ and
$\ell_{30}$.

Here we draw six horizontal lines with distances of $\frac 17$
between consecutive lines including the top and bottom. Due to
item heights, every type $\ell_{21}$ item intersects a line and we
associate it with one such line. Every type $\ell_{30}$ item
intersects at least two lines (as otherwise its height is at most
$\frac 27$) and we associate it with exactly two such lines.

As all widths are above $\frac 15$, every line intersects at most
four items, and it has at most four items associated with it. For
a given line (one of the six), let $y_{21}$ and $y_{30}$ be the
(integer) numbers of items of types $\ell_{21}$ and $\ell_{30}$
associated with it, where $\ell_{31}+\ell_{40} \leq 4$, as widths
are larger than $\frac 15$. We have $$1 \geq y_{21}(\frac
14+2^{30}\delta)+y_{30}(\frac 14-2^{22}\delta)
=(y_{21}+y_{30})/4+\delta(2^{30}y_{31}-2^{22}y_{40})  , $$ which
implies that either $y_{21}+y_{30}\leq 3$, or that $y_{30}=4$ and
$y_{21}=0$. The second option holds since in the case
$y_{21}+y_{30}\leq 3$ we get $y_{30} \geq 2^8 y_{21}$ similarly to
the proof of the previous lemma.

For every item associated with one line, we assign its weight to
the line, and for items associated with two lines, we assign half
of the weight to each such line, so the weight is split equally
between its two associated lines. Thus, an item of type
$\ell_{21}$ assigns a weight of $4$ to its line, and an item of
type $\ell_{30}$ assigns a weight of $3$ to each of its lines.

Consider a specific lines again. If $y_{21}+y_{30}\leq 3$, the
line is assigned at most a weight of $12$. In the case $y_{30}=4$
and $y_{21}=0$, it is also assigned a weight of $12$ (as the share
of weight for every item is $3$). As there are six lines, the
total weight is at most $6 \cdot 12=72$.

To prove the second bound, we consider types $\ell_{22}$ and
$\ell_{30}$. We draw lines and associate items as above. Every
line can intersect at most one item of type $\ell_{22}$ as the
width of such an item is above $\frac 12$. If a line does not have
such an item associated with it, it can have at most four
$\ell_{30}$ items associated with it. Otherwise, since
$w_{22}+2 \cdot w_{30}>1$, it can have at most one type
$\ell_{30}$ item associated with it. Let $x_0$ be the number of
lines without an $\ell_{22}$ item associated with them and let
$x_1=6-x_0$ be the number of lines having an $\ell_{22}$ item
associated with them. As every $\ell_{30}$ item is associated with
two lines, the number of $\ell_{30}$ items is at most
$$\left\lfloor \frac{1}{2} \cdot (4x_0+x_1)
\right\rfloor=\left\lfloor \frac{1}{2} \cdot (3x_0+6)
\right\rfloor=\left\lfloor \frac{3\cdot x_0}{2} \right\rfloor+3 .
$$ The number of $\ell_{30}$ items is also at most $8$, as their heights are above $\frac 13$,
and their widths are above $\frac 15$.

The number of $\ell_{22}$ items is $x_1$. Using the weights of
items ($8$ for type $\ell_{22}$ and $6$ for type $\ell_{30}$), the
total weight is at most
$$6 \cdot (3+\left\lfloor \frac{3x_0}{2}  \right\rfloor) +
8x_1=18+6\left\lfloor \frac{3x_0}{2}
\right\rfloor+8(6-x_0)=66+6\left\lfloor \frac{3x_0}{2}
\right\rfloor-8x_0 \leq 66+x_0  , $$ so for $x_0\leq 2$ the total
weight does not exceed $68$. The total weight is also at most $6
\cdot 8 +8x_1$, using the property that there are $x_1$ items of
type $\ell_{22}$ and at most eight items of type $\ell_{30}$, so
for $x_1 \leq 2$, the weight does not exceed $64$. The only
remaining case is $x_0=x_1=3$. In this case we have $\left\lfloor
\frac{3x_0}{2} \right\rfloor = 4$, and
$$66+6\left\lfloor \frac{3x_0}{2} \right\rfloor-8x_0=66+6\cdot 4 - 8
\cdot 3=66  . $$
\end{proof}

\begin{lemma}
We have $V_{1(k-1)}=126$ and $V_{1k}=112$.
\end{lemma}
\begin{proof}
We will consider type $V_{20}$ for all bounds, and type
$\ell_{1(k-1)}$ or $\ell_{1k}$ for the two bounds. We will use the
property that the width of type $\ell_{20}$ is above $\frac 15$.

$\frac 14 -2^{32}\delta > \frac 14 - \frac 1{2^{30}}> 0.24$, so
any horizontal line can intersect the interior of at most four
such items.

%
%

Recall that $$ w_{20} + 3 \cdot w_{1(k-1)} \geq 2 \cdot
 w_{20}
 + 2 \cdot w_{1(k-1)} \geq 3 \cdot w_{20} + w_{1(k-1)}
>1$$ and $$2 \cdot w_{20} +w_{1k} = 2 \cdot w_{20} +2 \cdot w_{1(k-1)}
>1  . $$ We also use the property  $2w_{1k}= 4 w_{1(k-1)} >1$.

Here, we draw $42$ lines of distances $\frac 1{43}$ between
consecutive lines, including the top and bottom. Since
$h_{20}=\frac 17+\eps$, every item of type $\ell_{20}$ contains
parts of at least six lines (as an item with at most five lines in
its interior has height at most $\frac{7}{43}$), and we associate
it with exactly six lines. Any item of the other type contains a
part of at least one line, and we associate it with one such line.


The calculation of $V_{1(k-1)}$ is as follows. A line can have at
most three $\ell_{1(k-1)}$ items associated with it due to the
width of this type that is larger than $\frac 14$. The maximum
number of $\ell_{20}$ items for a line with $3$, $2$, $1$, and $0$
such items can have at most the following numbers of $\ell_{20}$
items associated with it (respectively): $0$, $1$, $2$, and $4$.
Let $x_i$ be the number of lines for which the number of
$\ell_{1(k-1)}$ items associated with them is $i$, where
$x_0+x_1+x_2+x_3=42$.

As there are six lines associated with every $\ell_{20}$ item, we
have at most $\left\lfloor \frac 16 (x_2+2x_1+4x_0) \right\rfloor$
items of type $\ell_{20}$. The weight of an $\ell_{1(k-1)}$ item
is $1$, and the weight of a type $\ell_{20}$ item is $4$ (so the
share of every line associated with it is $\frac 23$). Thus the
total weight is at most
$$4 \left\lfloor \frac{x_2+2x_1+4x_0}6 \right\rfloor + (3x_3+2x_2+x_1) \leq
3(x_0+x_1+x_2+x_3)=3\cdot 42=126  . $$

The calculation of $V_{1k}$ is as follows. A line can have at most
one $\ell_{1k}$ item associated with it, as its width is above
$\frac 12$. The maximum number of $\ell_{20}$ items associated
with a line with one $\ell_{1k}$ item is one, and it there are no
$\ell_{1k}$ items, there can be at most four $\ell_{20}$ items
associated with the line, as their widths are above $\frac 15$.
Let $x_0$ and $x_1=42-x_0$ be the numbers of lines with no
 $\ell_{1k}$ items associated with them, and with one associated $\ell_{1k}$ item, respectively . As
there are six lines associated with every $\ell_{20}$ item, we
have at most $\left\lfloor \frac 16 (4x_0+x_1) \right\rfloor$ such
items. The weight of an $\ell_{1k}$ item is $2$. Thus the total
weight is at most $$ 4 \left\lfloor \frac {4x_0+x_1}6
\right\rfloor + 2x_1 \leq \frac 83(x_0+x_1)=\frac 83\cdot 42=112
. $$
\end{proof}

We conclude with the following theorem. \begin{theorem} The
asymptotic competitive ratio of any online algorithm for rectangle
packing is at least $\frac{1274}{667}\approx 1.9100449$.
\end{theorem}

\bibliographystyle{abbrv}
\bibliography{rectbib}

\end{document}